\documentclass{CSML}

\pdfoutput=1

\usepackage{lastpage}

\def\dOi{13(3:24)2017}
\lmcsheading%
{\dOi}
{1--\pageref{LastPage}}
{}
{}
{Aug.~\phantom09, 2016}
{Sep.~13, 2017}
{}

\usepackage[utf8]{inputenc} 

\usepackage{mathtools} 
\usepackage{bm}        
\usepackage{tikz}      
\usepackage{pgfplots}  
\usepackage{stackengine}
\usepackage[hidelinks]{hyperref}

\newcommand\barbelow[1]{\stackunder[1.2pt]{$#1$}{\rule{1.8ex}{.095ex}}}

\newcommand{\boweq}{\mathbin{\raisebox{0.1ex}{\barbelow{\bowtie}}}}
\newcommand{\bow}{\mathbin{\raisebox{0.1ex}{$\bowtie$}}}

\usepackage[linesnumbered, noend, noline, algoruled]{algorithm2e}
\SetArgSty{textup}     
\SetAlgorithmName{Procedure}{procedure}{List of Procedures}

\newcommand{\Z}{\mathbb{Z}}
\newcommand{\N}{\mathbb{N}}

\newcommand\post{\ensuremath{\textsf{Post}}}
\newcommand\posts{\ensuremath{\textsf{Post}^*}}
\newcommand\pre{\ensuremath{\textsf{Pre}}}
\newcommand\pres{\ensuremath{\textsf{Pre}^*}}

\renewcommand{\S}{\mathcal{S}}
\newcommand{\srun}{\xrightarrow{}}
\newcommand{\srunk}[1]{\xrightarrow{}\mathrel{\vphantom{\to}^{\!#1}}}

\newcommand\spost{\ensuremath{\post}}
\newcommand\sposts{\ensuremath{\post^*}}
\newcommand\spre{\ensuremath{\pre}}
\newcommand\spres{\ensuremath{\pre^*}}

\DeclareMathOperator{\upc}{\uparrow}
\DeclareMathOperator{\downc}{\downarrow}
\DeclareMathOperator{\upclex}{\uparrow_{\textsf{lex}}}
\DeclareMathOperator{\downclex}{\downarrow_{\textsf{lex}}}
\newcommand\ideals{\ensuremath{\textsf{Ideals}}}
\newcommand\idealslex{\ensuremath{\textsf{Ideals}_\textsf{lex}}}
\newcommand\idealdecomp{\ensuremath{\textsf{IdealDecomp}}}

\newcommand\leqlex{\ensuremath{\leq_\textsf{lex}}}
\newcommand\llex{\ensuremath{<_\textsf{lex}}}
\newcommand\glex{\ensuremath{>_\textsf{lex}}}
\newcommand\geqlex{\ensuremath{\geq_\textsf{lex}}}
\newcommand\preceqlex{\ensuremath{\preceq_\textsf{lex}}}

\newcommand{\C}{\ensuremath{\mathcal{C}}}

\newcommand{\problemx}[3]{
\par\noindent\underline{\sc#1}\par\nobreak\vskip.2\baselineskip
\begingroup\clubpenalty10000\widowpenalty10000
\setbox0\hbox{\bf INPUT:\ }\setbox1\hbox{\bf QUESTION:\ }
\dimen0=\wd0\ifnum\wd1>\dimen0\dimen0=\wd1\fi
\vskip-\parskip\noindent
\hbox to\dimen0{\box0\hfil}\hangindent\dimen0\hangafter1\ignorespaces#2\par
\vskip-\parskip\noindent
\hbox to\dimen0{\box1\hfil}\hangindent\dimen0\hangafter1\ignorespaces#3\par
\endgroup}

\newcommand{\defeq}{\stackrel{\scriptscriptstyle \text{def}}{=}}
\newcommand{\defiff}{\stackrel{\text{def}}{\iff}}
\newcommand{\tmi}{\textsf{Turing}_i}
\renewcommand{\vec}[1]{\bm{#1}}
\newcommand{\V}{\mathcal{V}}
\newcommand{\AT}{\textsf{AT}}

\begin{document}
\title{Well Behaved Transition Systems}

\author[M.~Blondin]{Michael Blondin\rsuper a}
\address{{\lsuper a}Université de Montréal, CNRS \& ENS Cachan – Université Paris-Saclay}
\email{blondimi@iro.umontreal.ca}  
\thanks{{\lsuper a}Supported by the Fonds de recherche du Québec – Nature et
  technologies (FRQNT), and the French Centre national de la recherche
  scientifique (CNRS)}

\author[A.~Finkel]{Alain Finkel\rsuper b}
\address{{\lsuper b}CNRS \& ENS Cachan – Université Paris-Saclay}
\email{finkel@lsv.ens-cachan.fr}  

\author[P.~McKenzie]{Pierre McKenzie\rsuper c}
\address{{\lsuper c}Université de Montréal}
\email{mckenzie@iro.umontreal.ca}  
\thanks{{\lsuper c}Supported by the ``Chaire Digiteo, ENS Cachan - {\'E}cole
  Polytechnique (France)'', and the Natural Sciences and Engineering
  Research Council of Canada.}


\begin{abstract}
  The well-quasi-ordering (i.e., a well-founded quasi-ordering such
  that all antichains are finite) that defines well-structured
  transition systems (WSTS) is shown not to be the weakest hypothesis
  that implies decidability of the coverability problem. We show
  coverability decidable for monotone transition systems that only
  require the absence of infinite antichains and call \emph{well
    behaved transitions systems} (WBTS) the new strict superclass of
  the class of WSTS that arises.  By contrast, we confirm that
  boundedness and termination are undecidable for WBTS under the usual
  hypotheses, and show that stronger monotonicity conditions can
  enforce decidability. Proofs are similar or even identical to
  existing proofs but the surprising message is that a hypothesis
  implicitely assumed minimal for twenty years in the theory of WSTS
  can meaningfully be relaxed, allowing more orderings to be handled
  in an abstract way.
\end{abstract}

\maketitle

\section{Introduction}\label{sec:intro}

The concept of a \emph{well-structured transition system} (WSTS) arose
thirty years ago, in 1987 precisely~\cite{F-icalp87,F90}, where such
systems were initially called \emph{structured transition systems} and
shown to have decidable termination and boundedness problems.  WSTS
were developed for the purpose of capturing properties common to a
wide range of formal models used in model-checking, system
verification and concurrent programming.  The coverability for such
systems was shown decidable in
1996~\cite{AbdullaCJT96,DBLP:journals/iandc/AbdullaCJT00}, thus
generalizing the decidability of coverability for lossy channel
systems \cite{DBLP:conf/lics/AbdullaJ93} but also generalizing a much
older result by Arnold and Latteux~\cite[Theorem 5, p.~391]{AL-79},
published in French and thus less accessible, stating that
coverability for vector addition systems in the presence of resets is
decidable. It is interesting to note that the algorithm used by Arnold
and Latteux in 1979 is an instance of the backward algorithm presented
in \cite{AbdullaCJT96} and applied to $\N^n$.

The usefulness of the WSTS stemmed from its clear abstract treatment
of the properties responsible for the decidability of coverability,
termination and boundedness.  This provided the impetus for an
intensive development of the theory of WSTS, begun in the
year 2000 (see
~\cite{Finkel&Schnoebelen:01,DBLP:journals/iandc/AbdullaCJT00} for
surveys and \cite{DBLP:conf/rta/BertrandDKSS12,DBLP:conf/gg/KonigS12,
  DBLP:conf/fossacs/WiesZH10,DBLP:conf/vmcai/ZuffereyWH12, esparza99,
  KaiserKW12,DBLP:conf/apn/GeeraertsHPR13} for a sample of recent applications of the WSTS).  WSTS
remain under development and are actively being investigated
\cite{Finkel&Goubault-Larrecq:09:a,Finkel&Goubault-Larrecq:09:b,DBLP:journals/jcss/GeeraertsRB06,SS-concur13,BS-fmsd2012,SS-icalp11}.

At its core, a WSTS is simply an infinite set $X$ (of states) with a
transition relation $\rightarrow\ \subseteq X \times X$. The set $X$
is quasi-ordered by $\leq$, and $\rightarrow$ fulfills one of various
possible monotonicities, i.e.\ compatibilities with $\leq$. The
quasi-ordering of $X$ is further assumed to be well,
i.e.\ well-founded and with no infinite antichains (see
Section~\ref{sect:prelim} for precise formal definitions).

Over the years, a number of strengthenings and weakenings of the
notion of monotonicity (of $\rightarrow$ w.r.t.\ $\leq$) were
introduced, with the goal of allowing WSTS to capture ever more
models~\cite{Finkel&Schnoebelen:01}. But to the best of our knowledge,
the wellness hypothesis attached to the quasi-ordering of $X$ was
never questioned, apparently under the assumption that wellness surely
ought to be the weakest possible hypothesis that would allow deducing
any form of decidability property.

Our main contribution is to prove the above assumption
unjustified. Indeed, we show that the wellness assumption in the
definition of WSTS can be relaxed while some decidabilities are
retained. More precisely, wellness in a quasi-ordering is equivalent
to the following two properties being fulfilled simultaneously:
\begin{itemize}
\item well-foundedness, i.e., the absence of an infinite
  descending sequence of elements, and
\item finiteness of antichains, i.e., the absence of infinite sets of
  pairwise incomparable elements.
\end{itemize}

\noindent We show that dropping well-foundedness from the definition of a WSTS
(resulting in a ``WBTS'') still allows deciding the
coverability problem, even in the presence of infinite
branching. Indeed, while the usual backward
algorithm~\cite{AbdullaCJT96} for coverability relies on
well-foundedness, the forward algorithm described here does not
require that property!

For example, the set $\Z$ of integers with increment and decrement as
its transitions defines a WBTS that is not a WSTS. Another example of
a WBTS that is not a WSTS is that of a vector addition system with
domain $\Z^d$ (hence without guards) rather than $\N^d$ and with
$d$-tuples ordered by building on the usual $\Z$-ordering
lexicographically rather than componentwise.  Yet a less artificial
example introduced in this paper is that of a \emph{weighted vector
  addition system}, defined as a normal $d$-VASS (over $\N^d$)
extended with a $\Z^w$-component ordered lexicographically (see
Sect.\ \ref{sect:cov} for precise definition and semantics).

Having defined WBTS, we argue that no general backward strategy would
apply to determine coverability for WBTS.  Our first contribution is
to nonetheless show the coverability problem for WBTS decidable, by
the use of a forward strategy.  Coverability is thus decidable for
each model mentioned in previous paragraph, sparing us the need for
separate independent arguments.

Deciding any computational problem, for a general class of WBTS,
naturally requires that the class verify a number of effectiveness
conditions.  One such condition in the case of coverability is the
need to be able to manipulate downward closed subsets of the system
domain.  Verifying this condition for weighted VASS requires an
analysis of the subsets of $\Z^d$ that are downward closed under the
lexicographical ordering. Elucidating the ideal structure of such
downward closed subsets of $\Z^d$ is our second contribution.

Our third contribution is to contrast WBTS and WSTS from the point of
view of the termination and boundedness problems.  As expected, under
monotonicity conditions that ensure decidability of termination and
boundedness for WSTS, we exhibit WBTS for which both problems are
undecidable.  By comparison, we investigate monotonicity conditions
that, even in WBTS, allow one to decide termination (in the finitely
branching case) and boundedness (in both the finitely and the
infinitely branching cases).

The paper is organized as follows. Section~\ref{sect:prelim}
introduces terminology. Section~\ref{sect:WBTS} defines well-behaved
transition systems, gives our first example of WBTS, defines
effectiveness and studies downward closed sets, including those of
$\Z^d$ under the lexicographical ordering.  Section~\ref{sect:cov}
proves coverability decidable for WBTS and defines the weighted VASS
model.  Section~\ref{sect:termination} compares the WSTS and the WBTS
from the point of view of the decidability of the termination and
boundedness problems.  Section~\ref{sect:conc} concludes with a
discussion and future work.

\section{Preliminaries}\label{sect:prelim}

\subsection{Orderings}

Let $X$ be a set and let $\leq\ \subseteq X \times X$. The relation
$\leq$ is a \emph{quasi-ordering} if it is reflexive and
transitive. If $\leq$ is additionally antisymmetric, then $\leq$ is a
\emph{partial order}. The set $X$ is \emph{well-founded (under
  $\leq$)} if there is no infinite strictly decreasing sequence $x_0 >
x_1 > \dots$ of elements of $X$. An \emph{antichain (under $\leq$)} is
a subset $A \subseteq X$ of pairwise incomparable elements, i.e.\ for
every $a, b \in A$, $a \not\leq b$ and $b \not\leq a$. We say that a
quasi-ordering $\leq$ is a \emph{well-quasi-ordering} for $X$ if $X$
is well-founded and contains no infinite antichain under $\leq$. Let
$A \subseteq X$, we define the \emph{downward closure} and
\emph{upward closure} of $A$ respectively as $\upc{A} \defeq \{x \in X
: x \geq a \text{ for some } a \in A\}$ and $\downc{A} \defeq \{x \in
X : x \leq a \text{ for some } a \in A\}$. A subset $A \subseteq X$ is
said to be \emph{downward closed} if $A = \downc{A}$ and \emph{upward
  closed} if $A = \upc{A}$. We say that a subset $B \subseteq A$, of
an upward closed set $A$, is a \emph{basis} of $A$ if $A =
\upc{B}$. An \emph{ideal} is a downward closed subset $I \subseteq X$
that is also \emph{directed}, i.e.\ it is nonempty and for every $a, b
\in I$, there exists $c \in I$ such that $a \leq c$ and $b \leq
c$. The set of ideals of $X$ is denoted $\ideals(X) \defeq \{I
\subseteq X : I = \downc{I} \text{ and } I \text{ is directed}\}$.

\subsection{Transition systems and effectiveness}

A \emph{transition system} is a pair $\S = (X, \srun)$ such that $X$
is a set whose elements are called the \emph{states} of $\S$, and a
\emph{transition relation} $\srun\ \subseteq X \times X$. We extend a
transition relation $\srun$ to $$\srunk{k}\ \defeq\ \underbrace{\srun
  \circ \srun \circ \dots \circ \srun}_{k \text{
    times}},\ \srunk{+}\ \defeq\ \bigcup_{k \geq 1} \srunk{k} \text{
  and } \srunk{*}\ \defeq\ \textsf{Id}\ \cup \srunk{+}$$ where
$\textsf{Id}$ is the identity relation. For every $x \in X$,
$\spost(x) \defeq \{y \in X : x \srun y\}$ and $\spre(x) \defeq \{y
\in X : y \srun x\}$ denote respectively the sets of \emph{immediate
  successors} and \emph{predecessors} of $x$. Similarly, for every $x
\in X$, $\sposts(x) \defeq \{y \in X : x \srunk{*} y\}$ and $\spres(x)
\defeq \{y \in X : y \srunk{*} x\}$ denote respectively the sets of
\emph{successors} and \emph{predecessors} of $x$. A transition system
is \emph{finitely branching} if $\spost(x)$ is finite for every state
$x$, otherwise it is \emph{infinitely branching}. An \emph{ordered
  transition system} $\S = (X, \srun, \leq)$ is a transition system
$(X, \srun)$ equipped with a quasi-ordering $\leq\ \subseteq X \times
X$. We naturally extend $\spost, \spre, \sposts$ and $\spres$ to
subsets of states, e.g.\ for $A \subseteq X$ we have $\sposts(A) =
\bigcup_{x \in A} \sposts(x)$.

A \emph{class $\C$ of transition systems} is any countable set of
transition systems. We denote the $i^\text{th}$ transition system of a
class $\C$, for some fixed enumeration, by $\C(i)$. For every class
$\C$ we require the existence of a set $\textsf{Enc}_\C \subseteq \N$
and a surjective \emph{representation map} $r : \textsf{Enc}_\C
\rightarrow \bigcup_i X_i$ where $X_i$ is the set of states of
$\C(i)$. Let $\textsf{Enc}_{X_i} = \{e \in \textsf{Enc}_\C : r(e) \in
X_i\}$, we further require the set $\{(i, e) : i \in \N, e \in
\textsf{Enc}_{X_i}\}$ to be decidable. A Turing machine $M$ over $\N
\times \N$ is said to \emph{compute} a relation $\rho \subseteq X_i
\times X_i$ if $M$ halts at least on $\textsf{Enc}_{X_i} \times
\textsf{Enc}_{X_i}$ and for each $e, e'\in \textsf{Enc}_{X_i}$, $M$
accepts $(e, e') \iff (r(e), r(e')) \in \rho$.

A class $\C$ of ordered transition systems is \emph{effective} if
there exists a pair of Turing machines $(M_{\rightarrow}, M_{\leq})$
operating on $\N \times \N \times \N$ such that, for each $i \in \N$,
$M_{\rightarrow}$ with first argument set to $i$ computes the
transition relation ``$\srun$'' of $\C(i)$ and $M_{\leq}$ with first
argument set to $i$ computes the ordering relation ``$\leq$'' of
$\C(i)$. We say that $\C$ is \emph{post-effective} if it is effective,
and if there exists an additional Turing machine that computes
$|\post_{\C(i)}(x)| \in \N \cup \{\infty\}$ on input $(i, x)$, with $i
\in \N$ and $x \in X_i$. Such a Turing machine, in combination with
$M_{\rightarrow}$, allows computing $\spost(x)$ whenever the latter is
finite. We say that $\C$ is \emph{upward pre-effective} if it is
effective, and if there exists an additional Turing machine that
computes a finite basis of $\upc{\pre_{\C(i)}(\upc{x})}$ on input $(i,
x)$, where $i \in \N$ and $x$ is a state of $\C(i)$. By extension, we
say that an ordered transition system $\S$ is \emph{effective}
(resp.\ \emph{post-effective}, \emph{upward pre-effective}) if the
degenerate class $\{\S\}$ is effective (resp.\ post-effective, upward
pre-effective).

Just as the states of an ordered transition system are encoded over
the natural numbers, we 
assume the existence of a representation map for ideals, and
that testing whether a natural number encodes an ideal under this map
is decidable.

\subsection{Monotone and well-structured transition systems}

Let $\S = (X, \srun, \leq)$ be an ordered transition system and
$\boweq \in \{\geq, \leq\}$. We say that $\S$ is \emph{(upward)
  monotone} if $\boweq$ is $\geq$ (resp.\ \emph{downward monotone} if
$\boweq$ is $\leq$) and if for every $x, x', y \in X$,
\begin{align}
  x \srun y \land x' \boweq x &\implies \exists y' \boweq y \text{ s.t. }
  x' \srunk{*} y'\ .\label{eq:monotonicity}
\end{align} We will consider variants of monotonicity
that were introduced in the literature by
modifying~(\ref{eq:monotonicity}) as follows:
\begin{align*}
  \text{transitive monotonicity:} && x \srun y \land x' \boweq x
  &\implies \exists y' \boweq y \text{ s.t. } x'
  \mathmakebox[17pt][l]{\srunk{+}} y' \\
  \text{strong monotonicity:} && x \srun y \land x' \boweq x &\implies
  \exists y' \boweq y \text{ s.t. } x' \mathmakebox[17pt][l]{\srun} y'
\end{align*}
Let $\bow \in \{>, <\}$ be the strict variant of $\boweq$.
For any one of the above monotonicities, an ordered
transition system is said to be \emph{strictly monotone} (with respect
to the relevant monotonicity) if it \emph{additionally} satisfies, for
every $x, x', y \in X$,
\begin{align*}
  x \srun y \land x' \bow x &\implies \exists y' \bow y \text{ s.t. }
  x' \srunk{\#} y'
\end{align*}
where $\# \in \{*, +, 1\}$ is in accord with the relevant
monotonicity. Note that strong monotonicity implies transitive
monotonicity which implies (standard) monotonicity.

\begin{defi}[{\cite{F90}}]\label{def:wsts}
  A \emph{well structured transition system (WSTS)} is a monotone
  transition system $S = (X, \srun, \leq)$ such that $X$ is
  well-quasi-ordered by $\leq$.
\end{defi}

The notion of downward monotonicity, perhaps
less known, has been introduced in~\cite{Finkel&Schnoebelen:01} to
study so-called \emph{downward WSTS} and has been used, for example,
to analyze timed alternating automata in~\cite{OW07}.

\section{Beyond WSTS: Well Behaved Transition Systems}\label{sect:WBTS}

We generalize well-structured transition systems by weakening the
well-quasi-ordering constraint. Instead, we consider monotone
transition systems ordered by quasi-orderings with no infinite
antichains. That is, we no longer require the ordering to be
well-founded:

\begin{defi}
  A \emph{well behaved transition system (WBTS)} is
  a monotone transition system $S = (X, \srun, \leq)$ such that $(X, \leq)$
  contains no infinite antichain.
\end{defi}

It is clear from the definition that every WSTS is a WBTS, however the
converse is not true. For example, consider automata that can increase
or decrease a single counter whose value ranges over $\Z$. Such
integer one-counter automata are readily seen to be WBTS, however they
are not WSTS since $\Z$ contains infinite strictly decreasing
sequences. WBTS can, in particular, be built from the (classical)
lexicographical ordering over finite words or integer tuples. These
orderings cannot be used in the setting of WSTS since they are not
well-founded, but are allowed in WBTS since these orderings do not
induce infinite antichains. WBTS are also closed under ordering
reversal, which is not the case of WSTS. More precisely, for an
ordered transition system $\S = (X, \srun, \leq)$, we define the
\emph{ordering reversal} of $\S$ as $\S^\textsf{ord-rev} \defeq (X,
\srun, \geq)$. It is easily seen that $\S$ is a WBTS with upward
monotonicity if, and only if, $\S^\textsf{ord-rev}$ is a WBTS with
downward monotonicity. In general, WSTS are not closed under ordering
reversal since the well-foundedness of an ordering is not necessarily
preserved when it is reversed, e.g.\ $\N$ is well-quasi-ordered by
$\leq$, but $0 < 1 < 2 < \dots$ is an infinite strictly decreasing
sequence over $\geq$.

\subsection{An example of WBTS}\label{sec:example:wbts}

As a proof of concept, and to build intuition, we exhibit a class of
WBTS that satisfies the monotonicities presented. This class is based
on integer vector addition systems with states that were recently
studied in~\cite{HH14,CHH14,BFGHM15}. An \emph{integer vector addition
  system with states} (\emph{$\Z^d$-VASS}) is a pair $\V = (Q, T)$
such that $Q$ and $T$ are finite sets, and $T \subseteq Q \times \Z^d
\times Q$ where $d > 0$. Sets $Q$ and $T$ are respectively called the
\emph{control states} and \emph{transitions} of $\V$. Intuitively, a
$\Z^d$-VASS is a vector addition systems with states (VASS), a model
equivalent to Petri nets, but in which the counters of the VASS may
drop below zero. Formally, a $\Z^d$-VASS induces a transition system
$(Q \times \Z^d, \srun)$ such that $(p, \vec{u}) \srun (q, \vec{v})
\defiff \exists (p, \vec{z}, q) \in T \text{ s.t. } \vec{v} = \vec{u}
+ \vec{z}$. The set of states $Q \times \Z^d$ of these systems is
typically ordered by equality on $Q$ and the usual componentwise
ordering of $\Z^d$, and are therefore neither well-founded nor without
infinite antichains. However, we show that $\Z^d$-VASS are WBTS when
ordered lexicographically, i.e.\ under $\preceqlex$ where $(p,
\vec{u}) \preceqlex (p', \vec{u}') \defiff p = p' \land \vec{u}
\leqlex \vec{u}'$ for $\leqlex$ the usual lexicographical ordering,
i.e.\ $\vec{u} \leqlex \vec{u}' \defiff \vec{u} = \vec{u}' \lor
\exists i \text{ s.t. } \vec{u}(i) < \vec{u}'(i) \land \forall j < i,
\vec{u}(j) = \vec{u}'(j)$.

\begin{prop}
  $\Z^d$-VASS ordered by $\preceqlex$ are WBTS with upward,
  downward, strong and strict monotonicity.
\end{prop}

\begin{proof}
  Let $\V = (Q, T)$ be a $\Z^d$-VASS. First note that any antichain of
  $Q \times \Z^d$ is of length at most $|Q|$. It remains to show that
  $\V$ is monotone. Let $(p, \vec{u}), (q, \vec{v}), (p', \vec{u}')
  \in Q \times \Z^d$ be such that $(p, \vec{u}) \srun (q, \vec{v})$
  and $(p, \vec{u}) \llex (p', \vec{u}')$. There exists $(p, \vec{z},
  q) \in T$ such that $\vec{v} = \vec{u} + \vec{z}$. By definition of
  $\preceqlex$, $p' = p$, hence $(p', \vec{u}') \srun (q, \vec{v}')$ for
  $\vec{v}' \defeq \vec{u}' + \vec{z}$. Let $1 \leq i \leq d$ be the
  smallest component such that $\vec{u}(i) \not= \vec{u}'(i)$. Since
  $\vec{u} \llex \vec{u}'$, we have $\vec{u}(i) < \vec{u}'(i)$, hence
  $\vec{v}(j) = \vec{u}(j) + \vec{z}(j) = \vec{u}'(j) + \vec{z}(j) =
  \vec{v}'(j)$ for every $1 \leq j < i$, and $\vec{v}(i) = \vec{u}(i)
  + \vec{z}(i) < \vec{u}'(i) + \vec{z}(i) = \vec{v}'(i)$. Therefore,
  $\vec{v} \llex \vec{v}'$ and consequently $(q, \vec{v}) \preceqlex
  (q, \vec{v}')$. Thus, $\V$ has upward, strong and strict
  monotonicity. Downward monotonicity follows symmetrically by
  considering $\glex$ instead of $\llex$.
\end{proof}

\subsection{Decomposition of downward closed sets into finite unions of ideals}

It was observed
in~\cite{Finkel&Goubault-Larrecq:09:a,FGL16,BFM14,BFM16} that any
downward closed subset of a well-quasi-ordered set is equal to a
finite union of ideals, which led to further applications in the study
of WSTS. Here we stress the fact that such finite decompositions also
exist in quasi-ordered sets with no infinite antichain. The existence
of such a decomposition has been proved numerous times (for partial
orderings instead of quasi-orderings) in the order theory
community~\cite{Bonnet,Pouzet79a,PZ85,F86,LMP87} under different
terminologies, and is a particular case of a more general set theory
result of Erdős \& Tarski~\cite{ET43} on the existence of \emph{limit
  numbers} between $\aleph_0$ and $2^{\aleph_0}$. We extract from
Bonnet~\cite{Bonnet} and Fra\"iss\'e ~\cite{F86} a simple proof
tailored to our situation. Specifically, our proof is based on the
fact that such decompositions exist in well-quasi-ordered sets and is
reminiscent of Fraïssé's proof strategy~\cite[Sect.~4.7.2,
  p.~124]{F86}, which is based on~\cite[Lemma~2, p.~193]{Bonnet}.

\begin{thm}[\cite{ET43,Bonnet,Pouzet79a,PZ85,F86,LMP87}]\label{thm:antichains:ideals}
  A countable quasi-ordered set $X$ contains no infinite antichain if,
  and only if, every downward closed subset of $X$ is equal to a
  finite union of ideals.
\end{thm}

\begin{proof}
  Let $X$ be a countable set quasi-ordered by $\leq$.

  \emph{Only if.} If $X$ is finite, the claim follows
  immediately. Suppose that $X$ is infinite and contains no infinite
  antichain. Let $D \subseteq X$ be a downward closed subset of $X$
  and let $D = \{d_0, d_1, \ldots\}$. We build a well-quasi-ordered
  subset $D' \subseteq D$. First, let us iteratively build a sequence
  of elements $(x_i)_{i \in \N}$ and a sequence of subsets $(D_i)_{i
    \in N}$. Let $D_0 \defeq D$ and $x_0 \defeq d_0$. For every $i >
  0$, let
  \begin{align*}
    D_i &\defeq D_{i-1} \setminus \downc{x_{i-1}}, \text{ and} \\
    x_i &\defeq d_j \text{ where $j$ is the smallest index such
      that $d_j \in D_i$}.
  \end{align*}
  Let $D' \defeq \{x_i : i \in \N \}$, let $\leq'$ be the
  quasi-ordering $\leq$ restricted to $D'$, and let $\downc'$ denote
  the downward closure under $\leq'$. We argue that $D'$ is
  well-quasi-ordered by $\leq'$. Recall that $(D', \leq')$ has no
  infinite antichain by hypothesis on $X$. We show that $(D', \leq')$
  is well-founded. By construction of $D'$, the following holds:
  \begin{align}
    x_i \not\leq x_j \text{ for every } i \in \N, j <
    i.\label{eq:not:leq}
  \end{align}
  Suppose that $D'$ contains an infinite strictly decreasing sequence:
  \begin{align}
    x_{i_0} > x_{i_1} > \ldots \label{eq:inf:desc}
  \end{align}
  where $i_j \neq i_k$ for every $j \neq k$. Since the set of indices
  $\{i_k : k \geq 0 \}$ is infinite, there necessarily exists an
  integer $k$ such that $i_0 < i_k$. Together with~(\ref{eq:not:leq}),
  this implies that $x_{i_0} \not\geq x_{i_k}$, which
  contradicts~(\ref{eq:inf:desc}).
  Therefore, $D'$ is well-founded under $\leq'$, which in turn implies
  that $D'$ is well-quasi-ordered by
  $\leq'$. By~\cite{F86,Finkel&Goubault-Larrecq:09:a,BFM16}, there
  exist $I_1, I_2, \ldots, I_k \in \ideals(D')$ such that $\downc'{D'}
  = I_1 \cup I_2 \cup \dots \cup I_k$.

  We claim that $D \subseteq \downc{D'}$, and hence that $D =
  \downc{D'}$. If $D = D'$, then the claim holds
  immediately. Otherwise, let $y \in D \setminus D'$. By construction
  of $D'$, $y < x_i$ for some $i \in \N$, and hence $y \in \downc{x_i}
  \subseteq \downc{D'}$. This implies that $D \subseteq
  \downc{D'}$.

  Therefore, $$D = \downc{D'} = \downc{(I_1 \cup I_2 \cup \dots \cup
    I_k)} = \downc{I_1} \cup \downc{I_2} \cup \dots \downc{I_k}\ .$$
  To conclude, it suffices to show that $\downc{I_i} \in \ideals(X)$
  for each $1 \leq i \leq k$. Obviously, $\downc{I_i}$ is downward
  closed, hence it suffices to show that it is directed. Let $a, b \in
  \downc{I_i}$, there exist $a', b' \in I_i$ such that $a \leq a'$ and
  $b \leq b'$. Since $I_i \in \ideals(D')$, there exists $c \in I_i$
  such that $a' \leq' c$ and $b' \leq' c$. Thus, $a \leq a' \leq c$
  and $b \leq b' \leq c$. Therefore, $\downc{I_i} \in \ideals(X)$ and
  we are done.

  \emph{If.} Conversely, suppose that there exists an infinite
  antichain $A \subseteq X$. We prove that there exists a downward
  closed subset $D \subseteq X$ that is not equal to a finite union of
  ideals. Let $D = \bigcup_{a \in A} \downc{a}$. Assume that there
  exist $I_1, I_2, \ldots, I_k \in \ideals(X)$ such that $D = I_1 \cup
  I_2 \cup \dots \cup I_k$. By the pigeonhole principle, there exists
  some $1 \leq i \leq k$ such that $I_i$ contains infinitely many
  elements from $A$. Let $a, b \in I_i$ be distinct elements. Since
  $I_i$ is directed, there exists $c \in I_i$ such that $a \leq c$ and
  $b \leq c$. Moreover, since $I_i \subseteq D$, there exists some $a'
  \in A$ such that $c \leq a'$. Thus, $a \leq a'$ and $b \leq
  a'$. Because $a$ and $b$ are distinct, at least two distinct
  elements of $A$ are comparable, i.e.\ either $a$ and $a'$, or $b$
  and $a'$. Therefore, $A$ is not an antichain, which is a
  contradiction, and hence $X$ has no infinite antichain.
\end{proof}

Let $X$ be a set quasi-ordered by an ordering $\leq$ having no
infinite antichain. Theorem~\ref{thm:antichains:ideals} allows us as
in~\cite{BFM14} to define a canonical finite decomposition of a
downward closed subset $D \subseteq X$, that is, the (finite) set
$\idealdecomp(D)$ of maximal ideals contained in $D$ under inclusion.

\subsection{Effectiveness of downward closed sets}

In this subsection, we describe effectiveness hypotheses that allow
manipulating downward closed sets in ordered transition systems.

\begin{defi}
  A class $\C$ of WBTS is \emph{ideally effective} if, given
  $\mathcal{S} = (X, \srun{}, \leq) \in \C$,
  \begin{itemize}
  \item the set of encodings of $\ideals(X)$ is recursive,
  \item the function mapping the encoding of a state $x \in X$ to the
    encoding of the ideal $\downc{x} \in \ideals(X)$ is computable;
  \item inclusion of ideals of $X$ is decidable;
  \item the downward closure $\downc{\spost(I)}$ expressed as a finite
    union of ideals is computable from the ideal $I \in \ideals(X)$.
  \end{itemize}
\end{defi}

\noindent Note that a class of WBTS is ideally effective if, and only if, the
class of its so-called completions~\cite{BFM14,BFM16} is
post-effective. The notion of completion naturally applies to WBTS,
but we do not use the notion in this paper.

Enforcing WBTS to be ideally effective is not an issue for all the
useful models of which we are aware. Indeed, a large scope of well
structured transition systems, hence of WBTS, are ideally effective
\cite{Finkel&Goubault-Larrecq:09:a}: Petri nets, VASS and their
extensions (with resets, transfers, affine functions), lossy channel
systems and extensions with data.

\definecolor{fillcolor}{rgb}{0.99,0.78,0.07}
\colorlet{fillcolor2}{cyan!80!blue}

\begin{figure}[h]
\begin{center}
  \begin{tikzpicture}
    \begin{axis}[xmin=-6.2, xmax=6.2, ymin=-6.2, ymax=6.65, xtick={-6,...,5}, ytick={-6,...,5}, yticklabels={,,}, xticklabels={,,}, axis y line=center, axis x line=middle, inner axis line style={black}, axis on top, thick, scale=0.7, transform shape]

      \addplot[black,fill=fillcolor, fill opacity=0.8, draw=none]
      coordinates {(1.2,-6.2) (1.2,6.2) (-6.2,6.2) (-6.2,-6.2)
        (1.2,-6.2)} \closedcycle;
      
      \addplot[mark=none, fillcolor, opacity=0.8, line width=4pt]
      coordinates {(2,3.2) (2,-6.2)};

      \foreach \x in {-6,...,1}{
        \foreach \y in {-6,...,6}{ 
          \addplot[mark=*, only marks, mark size=0.75pt, gray]
          coordinates {(\x,\y)};
        }
      }
      \foreach \y in {-6,...,3}{  
        \addplot[mark=*, only marks, mark size=0.75pt, gray]
        coordinates {(2,\y)};
      }
    \end{axis}
  \end{tikzpicture}\hspace{30pt}
  \begin{tikzpicture}
    \begin{axis}[xmin=-6.2, xmax=6.2, ymin=-6.2, ymax=6.65, xtick={-6,...,5}, ytick={-6,...,5}, yticklabels={,,}, xticklabels={,,}, axis y line=center, axis x line=middle, inner axis line style={black}, axis on top, thick, scale=0.7, transform shape]

      \addplot[black,fill=fillcolor, fill opacity=0.3, draw=none]
      coordinates {(1,-6.2) (1,6.2) (-6.2,6.2) (-6.2,-6.2) (1,-6.2)}
      \closedcycle;
      
      \addplot[mark=none, fillcolor, opacity=0.3, line width=4pt]
      coordinates {(2,3) (2,-6.2)};

      \addplot[black,fill=fillcolor2, fill opacity=0.9, draw=none]
      coordinates {(-1.8,-6.2) (-1.8,6.2) (-6.2,6.2) (-6.2,-6.2)
        (-1.8,-6.2)} \closedcycle;
      
      \addplot[mark=none, fillcolor2, opacity=0.9, line width=4pt]
      coordinates {(-1,5.2) (-1,-6.2)};

      \foreach \x in {-6,...,-2}{
        \foreach \y in {-6,...,6}{ 
          \addplot[mark=*, only marks, mark size=0.75pt, black!50!gray]
          coordinates {(\x,\y)};
        }
      }
      \foreach \y in {-6,...,5}{  
        \addplot[mark=*, only marks, mark size=0.75pt, black!50!gray]
        coordinates {(-1,\y)};
      }
    \end{axis}
  \end{tikzpicture}
\end{center}
\caption{Left: $\downclex{(2, 3)}$. Right: $\downclex{(-1, 5)} = \downclex{(2, 3)} + (-3, 2)$.}\label{fig:down:lex}
\end{figure}

As an example, we argue that $\Z^d$-VASS introduced in
Sect.~\ref{sec:example:wbts} form an ideally effective class of
WBTS. To do so, we need to investigate the downward and upward closed
sets of $\Z^d$ under $\preceqlex$. Since the control states are
ordered under equality, we may only consider $\leqlex$. Let $\vec{x} \in
\Z^d$, we give descriptions of $\downclex{\vec{x}}$ and
$\upclex{\vec{x}}$ where $\downclex$ and $\upclex$ denote respectively
the downward and upward closures under $\leqlex$. Let $\textsf{down}_d
: \Z^d \to 2^{\Z^d}$ be defined as follows
$$
\textsf{down}_d(x_1, x_2, \ldots, x_d) =
\begin{cases}
  \downc(x_1-1) \times \Z^{d-1} \cup \{x_1\} \times
  \textsf{down}_{d-1}(x_2, x_3, \ldots, x_d) & \text{if } d > 1
  \\ \downc{x_1} & \text{if } d = 1
\end{cases}
$$ and let $\textsf{up}_d : \Z^d \to 2^{\Z^d}$ be defined in the same
way by replacing $\downc$ with $\upc$. We have $\downclex{\vec{x}} =
\textsf{down}_d(\vec{x})$ and $\upclex{\vec{x}} =
\textsf{up}_d(\vec{x})$. For example, $\downclex(2, 3) = (\downc{1}
\times \Z) \cup (\{2\} \times \downc{3})$ is depicted on the left of
Fig.~\ref{fig:down:lex}.

In order to describe the ideals of $\Z^d$ under $\leqlex$, denoted
$\idealslex(\Z^d)$, we first make the following observation on
downward closed subsets:
\begin{prop}\label{prop:downward:ideal}
  Let $D \subseteq \Z^d$. If $D$ is downward closed under $\leqlex$,
  then $D \in \idealslex(\Z^d)$.
\end{prop}

\begin{proof}
  Let $\vec{u}, \vec{v} \in D$. Since $\leqlex$ is total, $\vec{u}
  \leqlex \vec{v}$ or $\vec{v} \leqlex \vec{u}$, and hence $D$ is
  directed.
\end{proof}

$\idealslex(\Z^d)$ can be described as follows:

\begin{prop}\label{prop:ideals:description}
  $\idealslex(\Z^d) = X_d$ where
  $$X_d =
  \begin{cases}
    \left\{\Z^d\right\} \cup \left\{\downc(x-1) \times \Z^{d-1} \cup
    \{x\} \times I : x \in \Z, I \in X_{d-1}\right\} & \text{if $d >
      1$,} \\ \left\{\Z\right\}\phantom{{}^d} \cup \left\{\downc{x} : x
    \in \Z\right\} & \text{if $d = 1$.}
  \end{cases}
  $$
\end{prop}

\begin{proof}
  We proceed by induction on $d$. The base case is immediate since
  $\leqlex$ coincides with $\leq$ for $d=1$. Let $d > 1$, suppose the
  claim holds for $d-1$.\medskip

  \noindent{``$\subseteq$''}: Let $I \in \idealslex(\Z^d)$ and let us
  show that $I \in X_d$. Let $F = \{\vec{v}(1) : \vec{v} \in I\}$. If
  $F$ is unbounded from above, then $I = \Z^d$ and trivially $I \in
  X_d$. Otherwise, let $x$ be the largest element of $F$. By downward
  closure of $I$ under $\leqlex$ and by definition of $\leqlex$,
  \begin{align}
    I &= \{\vec{v} \in \Z^d : \vec{v}(1) < x\} \cup \{[x\ \vec{u}] \in
    I : \vec{u} \in \Z^{d-1}\}\ . \label{eq:ideals:rec}
  \end{align}
  Let us show that $I' \defeq \{\vec{u} \in \Z^{d-1} : [x\ \vec{u}]
  \in I\}$ is downward closed under $\leqlex$, and hence that $I' \in
  \ideals(\Z^{d-1})$ by Prop.~\ref{prop:downward:ideal}. Let $\vec{u}
  \in I'$ and $\vec{u}' \leqlex \vec{u}$. We have $[x\ \vec{u}']
  \leqlex [x\ \vec{u}]$, hence $[x\ \vec{u}'] \in I$ by downward
  closure of $I$ under $\leqlex$, and thus $\vec{u}' \in
  I'$. Therefore, by~(\ref{eq:ideals:rec}), we have
  \begin{alignat*}{2}
    I &= \{\vec{v} \in \Z^d : \vec{v}(1) < x\} &&\cup \{[x\ \vec{u}] :
    \vec{u} \in I'\} \\
    &= \downc{(x-1)} \times \Z^{d-1} &&\cup \{x\} \times I'\ .
  \end{alignat*}
  By induction hypothesis, $I' \in X_{d-1}$, and thus $I \in
  X_d$.\medskip

  \noindent{``$\supseteq$''}: Let $I \in X_d$ and let us show that $I
  \in \idealslex(\Z^d)$. If $I = \Z^d$, then $I \in
  \idealslex(\Z^d)$. Assume that $I \not= \Z^d$, then by definition of
  $X_d$ and by induction hypothesis
  \begin{align*}
    I &= \downc(x-1) \times \Z^{d-1} \cup \{x\} \times I'
  \end{align*}
  for some $x \in \Z$ and $I' \in \idealslex(\Z^{d-1})$. Let us show
  that $I$ is downward closed under $\leqlex$, and hence that $I \in
  \idealslex(\Z^d)$ by Prop.~\ref{prop:downward:ideal}. Let $\vec{v}
  \in I$ and $\vec{v}' \leqlex \vec{v}$. If $\vec{v} \in \downc(x-1)
  \times \Z^{d-1}$, then $\vec{v}' \in \downc(x-1) \times \Z^{d-1}
  \subseteq I$ since $\vec{v}'(1) \leq \vec{v}(1)$. If $\vec{v} \in
  \{x\} \times I'$, then there are two cases to consider:

  \begin{itemize}
  \item If $\vec{v}'(1) < \vec{v}(1)$, then $\vec{v}'(1) \in
    \downc(x-1) \times \Z^{d-1} \subseteq I$.\\[-10pt]
  \item If $\vec{v}'(1) = \vec{v}(1)$, then there exist $\vec{u} \in
    I'$ and $\vec{u}' \in \Z^{d-1}$ such that $\vec{v} =
    [x\ \vec{u}]$, $\vec{v}' = [x\ \vec{u}']$ and $\vec{u}' \leqlex
    \vec{u}$. Since $I'$ is downward closed under $\leqlex$, we have
    $\vec{u}' \in I'$. Therefore, $\vec{v}' = [x\ \vec{u}'] \in \{x\}
    \times I' \subseteq I$. \qedhere
  \end{itemize}
\end{proof}

\noindent Ideals of $\Z^d$ can be categorized into $d+1$ types, as illustrated
in Fig.~\ref{fig:ideals:types} for $d=2$. By
Prop.~\ref{prop:downward:ideal}, ideals of $\Z^d$ under $\leqlex$ are
precisely the downward closed sets under $\leqlex$. Symmetrically,
ideals of $\Z^d$ under $\geqlex$ are the downward closed sets under
$\geqlex$, which in turn are the upward closed sets under
$\leqlex$. Therefore, upward closed subsets of $\Z^d$ under $\leqlex$
can be described by replacing $\upc$ with $\downc$ in the description
of $\idealslex(\Z^d)$ given by
Prop.~\ref{prop:ideals:description}. Upward and downward closed
subsets can thus be represented symbolically with disjoint finite
unions of products of terms of the form $\{a\}$, $\downc{a}$ or
$\upc{a}$, and $\Z$.

\begin{figure}[h]
\begin{center}
  \begin{tikzpicture}
    \begin{axis}[xmin=-6.2, xmax=6.95, ymin=-6.2, ymax=6.95, xtick={-6,...,5}, ytick={-6,...,5}, yticklabels={,,}, xticklabels={,,}, axis y line=center, axis x line=middle, inner axis line style={black}, axis on top, thick, scale=0.45, transform shape]

      \addplot[black,fill=fillcolor, fill opacity=0.8, draw=none]
      coordinates {(2.2,-6.2) (2.2,6.2) (-6.2,6.2) (-6.2,-6.2)
        (2.2,-6.2)} \closedcycle;
      
      \addplot[mark=none, fillcolor, opacity=0.8, line width=4pt]
      coordinates {(3,-1.8) (3,-6.2)};

      \foreach \x in {-6,...,2}{
        \foreach \y in {-6,...,6}{ 
          \addplot[mark=*, only marks, mark size=0.45pt, gray]
          coordinates {(\x,\y)};
        }
      }
      \foreach \y in {-6,...,-2}{  
        \addplot[mark=*, only marks, mark size=0.45pt, gray]
        coordinates {(3,\y)};
      }
    \end{axis}
  \end{tikzpicture}\hspace{30pt}
  \begin{tikzpicture}
    \begin{axis}[xmin=-6.2, xmax=6.95, ymin=-6.2, ymax=6.95, xtick={-6,...,5}, ytick={-6,...,5}, yticklabels={,,}, xticklabels={,,}, axis y line=center, axis x line=middle, inner axis line style={black}, axis on top, thick, scale=0.45, transform shape]

      \addplot[black,fill=fillcolor, fill opacity=0.8, draw=none]
      coordinates {(2.2,-6.2) (2.2,6.2) (-6.2,6.2) (-6.2,-6.2)
        (2.2,-6.2)} \closedcycle;
      
      \foreach \x in {-6,...,2}{
        \foreach \y in {-6,...,6}{ 
          \addplot[mark=*, only marks, mark size=0.45pt, gray]
          coordinates {(\x,\y)};
        }
      }
    \end{axis}
  \end{tikzpicture}\hspace{30pt}
  \begin{tikzpicture}
    \begin{axis}[xmin=-6.2, xmax=6.95, ymin=-6.2, ymax=6.95, xtick={-6,...,5}, ytick={-6,...,5}, yticklabels={,,}, xticklabels={,,}, axis y line=center, axis x line=middle, inner axis line style={black}, axis on top, thick, scale=0.45, transform shape]

      \addplot[black,fill=fillcolor, fill opacity=0.8, draw=none]
      coordinates {(-6.2,-6.2) (-6.2,6.2) (6.2,6.2) (6.2,-6.2)
        (-6.2,-6.2)} \closedcycle;
      
      \foreach \x in {-6,...,6}{
        \foreach \y in {-6,...,6}{ 
          \addplot[mark=*, only marks, mark size=0.45pt, gray]
          coordinates {(\x,\y)};
        }
      }
    \end{axis}
  \end{tikzpicture}
\end{center}
\caption{Example of each of the three types of ideals of $\Z^2$ under
  lexicographical ordering.}\label{fig:ideals:types}
\end{figure}

Inclusion between two downward (resp.\ upward) closed subsets is
decidable, e.g.\ we may translate $I \subseteq J$ into a formula of
the first-order theory of integers with addition,
i.e.\ $\textsf{FO}(\Z, +, <)$, which is decidable~\cite{Pr29}. For
example, to test whether $(\downc{2} \times \Z) \subseteq (\downc{1}
\times \Z) \cup (\{2\} \times \downc{3})$, we verify if the following
formula is satisfiable: $\varphi \defeq \forall x, y \in \Z,\ (x \leq
2) \implies ((x \leq 1) \lor (x = 2 \land y \leq 3))$.

Moreover, we can effectively add some $\vec{z} \in \Z^d$ to a downward
(resp.\ upward) closed subset $A \subseteq \Z^d$. This can be done in
polynomial time by adding $\vec{z}$ to the ``maximal points'' of the
representation of $A$. For example, $\downclex{(2, 3)} + (-3, 2) =
(\downc{1} \times \Z) \cup (\{2\} \times \downc{3}) + (-3, 2) =
(\downc(-2) \times \Z) \cup (\{-1\} \times \downc{5}) = \downclex{(-1,
  5)}$ as illustrated on the right of Fig.~\ref{fig:down:lex}.

From these observations, we can encode and manipulate downward/upward
closed subsets effectively, and thus:

\begin{prop}
  $\Z^d$-VASS form a post-effective and ideally effective class of
  WBTS.
\end{prop}

\section{Decidability of Coverability for Well Behaved Transition Systems}\label{sect:cov}

The coverability problem is defined as follows: on input an ordered
transition system $\S = (X, \srun, \leq)$ and two states $x, y \in X$,
determine whether $x \srunk{*} y'$ for some $y' \geq y$. In this
section, we show coverability decidable for WBTS that enjoy the
so-called ideal effectiveness. With all the effectiveness notions in
place, we then define the (apparently new) notion of a $(d,w)$-VASS,
i.e., \emph{weighted} $d$-VASS, as a WBTS that fulfills the required
effectiveness and thus has a decidable coverability problem.

The \emph{backward
  algorithm}~\cite{AbdullaCJT96,DBLP:conf/lics/AbdullaJ93,AL-79} is
perhaps the best known algorithm for deciding coverability in upward
pre-effective WSTS. It proceeds by starting with $\upc{y}$ and
computing iteratively the sequence $\upc{\pre(\upc{y})},
\upc{\pre(\upc{\pre(\upc{y})})}, \dots$ until the union of this
sequence stabilizes, which is guaranteed to happen by $\leq$ being a
well-quasi-ordering. The finite union of this sequence yields
$\upc{\pres(y)}$, and hence it suffices to verify whether this
contains $x$ or not. When $\leq$ is not a well-quasi-ordering, this
approach fails since the procedure may never halt. For example,
consider the $\Z^2$-VASS $\V \defeq (\{q\}, \{(q, (0, 1), q)\}$
ordered lexicographically. Since $\Z^d$-VASS are upward pre-effective,
we may execute the backward algorithm on $\V$. To verify whether
$\vec{y} = (1, 1)$ is coverable from $\vec{x} = (0, 0)$, the backward
algorithm iteratively computes $(q, \upclex{(1, 1)}), (q, \upclex{(1,
  0)), (q, \upclex{(1, -1)}), (q, \upclex{(1, -2)}}), \dots$ as
illustrated in Fig.~\ref{fig:backward:2zvass}. Since this sequence is
strictly increasing and does not contain $\vec{x}$, the backward
algorithm never halts.

\begin{figure}[h]
\begin{center}
  \begin{tikzpicture}
    \begin{axis}[xmin=-1.2, xmax=6.65, ymin=-6.2, ymax=6.2, xtick={-6,...,5}, ytick={-6,...,5}, yticklabels={,,}, xticklabels={,,}, axis y line=center, axis x line=middle, inner axis line style={black}, axis on top, thick, scale=0.45, transform shape]

      \addplot[black,fill=fillcolor, fill opacity=0.8, draw=none]
      coordinates {(1.8,-6.2) (1.8,6.2) (6.2,6.2) (6.2,-6.2)
        (1.8,-6.2)} \closedcycle;
      
      \addplot[mark=none, fillcolor, opacity=0.8, line width=4pt]
      coordinates {(1,0.8) (1,6.2)};

      \foreach \x in {2,...,6}{
        \foreach \y in {-6,...,6}{ 
          \addplot[mark=*, only marks, mark size=0.75pt, gray]
          coordinates {(\x,\y)};
        }
      }
      \foreach \y in {1,...,6}{  
        \addplot[mark=*, only marks, mark size=0.75pt, gray]
        coordinates {(1,\y)};
      }

      \addplot[mark=square*, only marks, mark size=3pt, black] coordinates
              {(0,0)};
      \addplot[mark=triangle*, only marks, mark size=3pt, black]
      coordinates {(1,1)};
    \end{axis}
  \end{tikzpicture}\hspace{7pt}
  \begin{tikzpicture}
    \begin{axis}[xmin=-1.2, xmax=6.65, ymin=-6.2, ymax=6.2, xtick={-6,...,5}, ytick={-6,...,5}, yticklabels={,,}, xticklabels={,,}, axis y line=center, axis x line=middle, inner axis line style={black}, axis on top, thick, scale=0.45, transform shape]

      \addplot[black,fill=fillcolor, fill opacity=0.8, draw=none]
      coordinates {(1.8,-6.2) (1.8,6.2) (6.2,6.2) (6.2,-6.2)
        (1.8,-6.2)} \closedcycle;
      
      \addplot[mark=none, fillcolor, opacity=0.8, line width=4pt]
      coordinates {(1,-0.2) (1,6.2)};

      \foreach \x in {2,...,6}{
        \foreach \y in {-6,...,6}{ 
          \addplot[mark=*, only marks, mark size=0.75pt, gray]
          coordinates {(\x,\y)};
        }
      }
      \foreach \y in {0,...,6}{  
        \addplot[mark=*, only marks, mark size=0.75pt, gray]
        coordinates {(1,\y)};
      }
      \addplot[mark=square*, only marks, mark size=3pt, black] coordinates
              {(0,0)};
      \addplot[mark=triangle*, only marks, mark size=3pt, black]
      coordinates {(1,1)};
    \end{axis}
  \end{tikzpicture}\hspace{7pt}
  \begin{tikzpicture}
    \begin{axis}[xmin=-1.2, xmax=6.65, ymin=-6.2, ymax=6.2, xtick={-6,...,5}, ytick={-6,...,5}, yticklabels={,,}, xticklabels={,,}, axis y line=center, axis x line=middle, inner axis line style={black}, axis on top, thick, scale=0.45, transform shape]

      \addplot[black,fill=fillcolor, fill opacity=0.8, draw=none]
      coordinates {(1.8,-6.2) (1.8,6.2) (6.2,6.2) (6.2,-6.2)
        (1.8,-6.2)} \closedcycle;
      
      \addplot[mark=none, fillcolor, opacity=0.8, line width=4pt]
      coordinates {(1,-1.2) (1,6.2)};

      \foreach \x in {2,...,6}{
        \foreach \y in {-6,...,6}{ 
          \addplot[mark=*, only marks, mark size=0.75pt, gray]
          coordinates {(\x,\y)};
        }
      }
      \foreach \y in {-1,...,6}{  
        \addplot[mark=*, only marks, mark size=0.75pt, gray]
        coordinates {(1,\y)};
      }
      \addplot[mark=square*, only marks, mark size=3pt, black] coordinates
              {(0,0)};
      \addplot[mark=triangle*, only marks, mark size=3pt, black]
      coordinates {(1,1)};
    \end{axis}
  \end{tikzpicture}\hspace{7pt}
  \begin{tikzpicture}
    \begin{axis}[xmin=-1.2, xmax=6.65, ymin=-6.2, ymax=6.2, xtick={-6,...,5}, ytick={-6,...,5}, yticklabels={,,}, xticklabels={,,}, axis y line=center, axis x line=middle, inner axis line style={black}, axis on top, thick, scale=0.45, transform shape]

      \addplot[black,fill=fillcolor, fill opacity=0.8, draw=none]
      coordinates {(1.8,-6.2) (1.8,6.2) (6.2,6.2) (6.2,-6.2)
        (1.8,-6.2)} \closedcycle;
      
      \addplot[mark=none, fillcolor, opacity=0.8, line width=4pt]
      coordinates {(1,-2.2) (1,6.2)};

      \foreach \x in {2,...,6}{
        \foreach \y in {-6,...,6}{ 
          \addplot[mark=*, only marks, mark size=0.75pt, gray]
          coordinates {(\x,\y)};
        }
      }
      \foreach \y in {-2,...,6}{  
        \addplot[mark=*, only marks, mark size=0.75pt, gray]
        coordinates {(1,\y)};
      }
      \addplot[mark=square*, only marks, mark size=3pt, black] coordinates
              {(0,0)};
      \addplot[mark=triangle*, only marks, mark size=3pt, black]
      coordinates {(1,1)};
    \end{axis}
  \end{tikzpicture}
\end{center}
\caption{From left to right: first four iterations of the backward
  algorithm trying to determine whether $\vec{y} = (1, 1)$ is
  coverable from $\vec{x} = (0, 0)$ in the $\Z^2$-VASS $\V = (\{q\},
  \{(q, (0, 1), q)\})$. Vectors $\vec{x}$ and $\vec{y}$ are
  respectively marked as~\protect\tikz\fill (0.1, 0.1) rectangle (0.3,
  0.3); and~\protect\tikz\fill (0,0) -- (0.1,0.2) --
  (0.2,0);.}\label{fig:backward:2zvass}
\end{figure}

By contrast, we show in this section that the forward approach for
coverability, initially presented by Geeraerts, Raskin, and Van Begin
\cite{DBLP:conf/fsttcs/GeeraertsRB04,DBLP:journals/jcss/GeeraertsRB06} for WSTS\footnote{The idea had
  also appeared in 1982; see~\cite[Corollary
    8.7]{Pa82} where it was applied to the
  reachability problem for communicating finite automata with FIFO
  channels.} and simplified in~\cite{Finkel&Goubault-Larrecq:09:b,BFM14,BFM16}, avoids this problem
and actually works for WBTS under the same effectiveness
hypothesis. The approach relies on decompositions of downward closed
sets into finitely many ideals. The proof that the forward approach of
~\cite{BFM14,BFM16} is correct for WSTS requires no essential
modification for WBTS, but we expand it in more details here.

In order to decide whether $y$ is coverable from $x$, we execute two
procedures in parallel, one looking for a coverability certificate and
one looking for a non coverability certificate. Procedure~\ref{alg:1}
iteratively computes
\begin{align*}
  \downc{x},\ \downc{\post(\downc{x})},\ \downc{\post(\downc{\post(\downc{x}))}},\ \dots
\end{align*}
until it finds $y$.

\begin{algorithm}[h]
  \DontPrintSemicolon
  $D \leftarrow \downc{x}$\;
  \While{$y \not\in D$}{
    $D \leftarrow D \cup \downc{\post(D)}$\;
    }
  \Return{\textit{true}}
  \caption{searches for a coverability certificate of $y$ from
    $x$}\label{alg:1}
\end{algorithm}

\begin{algorithm}[h]
  \DontPrintSemicolon
  $i \leftarrow 0$\;
  \While{$\neg(\downc{\post(D_i)} \subseteq D_i$ and $x \in D_i$ and $y \not\in D_i)$}{
    $i \leftarrow i + 1$\;
    }
  \Return{\textit{false}}
  \caption{enumerates inductive invariants to find non coverability
    certificate of $y$ from $x$.}\label{alg:2}
\end{algorithm}

The second procedure enumerates inductive invariants in some fixed
order $D_1, D_2, \dots$, i.e.\ downward closed subsets $D_i \subseteq
X$ such that $\downc{\post(D_i)} \subseteq D_i$. Any inductive
invariant $D_i$ such that $x \in D_i$ and $y \not\in D_i$ is a
certificate of non coverability. This is due to the fact that every
inductive invariant $D_i$ is an ``over-approximation'' of $\posts(x)$
if it contains $x$. Moreover, by standard monotonicity,
$\downc{\posts(x)}$ is such an inductive invariant and may eventually
be found.

We show that these two procedures are correct:

\begin{thm}\label{thm:cover:wsts:algo} 
  Let $\S = (X, \srun, \leq)$ be a WBTS, and let $x, y \in X$.
  \begin{enumerate}
  \item $y$ is coverable from $x$ if, and only if,
    Procedure~\ref{alg:1} terminates.
  \item $y$ is not coverable from $x$ if, and only if,
    Procedure~\ref{alg:2} terminates.
  \end{enumerate}
\end{thm}

\begin{proof}\leavevmode
  \begin{enumerate}
  \item Procedure~\ref{alg:1} computes
    $$ D = \bigcup_{k=0}
    \underbrace{\downc\post(\cdots\ \downc\post(\downc{x}))}_{k \text{
        times}}\ .
    $$ It suffices to show that $D = \downc{\posts(x)}$, since $y$ is
    coverable from $x$ if, and only if, $y \in
    \downc{\posts(x)}$.

    The inclusion $\downc{\posts(x)} \subseteq D$ is immediate. Let us
    prove that $D \subseteq \downc{\posts(x)}$. Let $z \in D$. There
    exist $k \in \N$ and $x_0, x_0', x_1, x_1', \ldots, x_k, x_k'$
    such that $x_0 = x$, $x_k' = z$, $x_i \geq x_i'$ for every $0 \leq
    i \leq k$, and $x_i' \srun{} x_{i+1}$ for every $0 \leq i < k$. By
    applying monotonicity $k$ times, we obtain $x \srunk{*} z'$ for
    some $z' \geq z$. Thus, $z \in \downc{\posts(x)}$, whence $D
    \subseteq \downc{\posts(x)}$.

  \item By a simple induction, it can be shown that
    $\downc{\sposts(D)} \subseteq D$ for every inductive invariant
    $D$. If Procedure~\ref{alg:2} terminates, then $y \not\in D
    \supseteq \downc{\sposts(D)} \supseteq \downc{\sposts(x)}$ which
    implies that $y$ is not coverable from $x$.

    It remains to show that Procedure~\ref{alg:2} terminates whenever
    $y$ is not coverable from $x$. To do so, it suffices to prove that
    $\downc{\sposts(x)}$ is an inductive invariant. Indeed, this
    implies that $\downc{\sposts(x)}$ is eventually found by
    Procedure~\ref{alg:2} when $y$ is not coverable from
    $x$. Formally, let us show that $\downc{\post(\downc{\posts(x)})}
    \subseteq \downc{\sposts(x)}$. Let $b \in
    \downc{\post(\downc{\posts(x)})}$, there exists $a', a, b'$ such
    that $x \srunk{*} a'$, $a' \geq a$, $a \srun b'$ and $b' \geq
    b$. By monotonicity, there exists $b'' \geq b'$ such that $a'
    \srunk{*} b''$. Therefore, $x \srunk{*} b''$ and $b' \geq b$,
    hence $b \in \downc{\posts(x)}$.\qedhere
  \end{enumerate}
\end{proof}

\noindent In order to implement Procedure~\ref{alg:1} and Procedure~\ref{alg:2},
some effectiveness hypotheses must be made. We argue that both
procedures may be implemented for ideally effective classes of
WBTS. We first need the following crucial proposition concerning
inclusion of ideals, in particular for testing inclusion of downward
closed sets. We include its proof for completeness:

\begin{prop}[{\cite{BFM14,BFM16}}] \label{prop:inclusdansunion}
  Let $X$ be a quasi-ordered set. For every $I, J_1, J_2, \ldots,
  \linebreak J_m \in \ideals(X)$, $I \subseteq J_1 \cup J_2 \cup \dots
  \cup J_m$ if, and only if, $I \subseteq J_j$ for some $1 \leq j \leq
  m$.
\end{prop}

\begin{proof}
  We claim that if a directed set $I$ is included in $J \cup K$ where
  $J$ and $K$ are downward closed, then either $I \subseteq J$ or $I
  \subseteq K$. The claim implies the proposition by a straightforward
  induction since an ideal is directed and any union of ideals is
  downward closed.

  To see the claim, let $I \subseteq J \cup K$ under the conditions
  stated and suppose to the contrary that there exist $s \in I
  \setminus J$ and $t \in I \setminus K$. Since $I$ is directed, there
  exists $u \in I$ such that $s \leq u$ and $t \leq u$. Since $u \in
  I$, either $u \in J$ or $u \in K$.  By downward closures of $J$ and
  $K$, either $s \in J$ or $t \in K$, a contradiction that proves the
  claim.
\end{proof}

From the definition of ideally effective classes of WBTS and from
Prop.~\ref{prop:inclusdansunion}, we can show that the elementary
operations of Procedure~\ref{alg:1} and Procedure~\ref{alg:2} are
computable. Formally:

\begin{lem}\label{lem:down:effectiveness}
  Let $\mathcal{C}$ be an ideally effective class of WBTS. There exist
  Turing machines $(M_{\textsf{down}}, M_{\subseteq},
  M_{\downc{\!\textsf{Post}}}, M_{\textsf{memb}})$ such that, on input
  $\S = (X, \srun, \leq) \in \mathcal{C}$,
  \begin{enumerate} 
  \item $M_{\textsf{down}}$ enumerates every downward closed subsets
    of $X$ by their ideal decomposition,
  \item $M_{\subseteq}$ decides inclusion between downward closed
    subsets of $X$ prescribed by their ideal decomposition,
  \item $M_{\downc{\!\textsf{Post}}}$ computes the ideal decomposition
    of $\downc{\post(D)}$ for downward closed subsets $D$ of $X$
    prescribed by their ideal decomposition,
  \item $M_{\textsf{memb}}$ decides $x \in D$, given $x \in X$ and a
    downward closed subset $D \subseteq X$ prescribed by its ideal
    decomposition.
  \end{enumerate}
\end{lem}

\begin{proof}\leavevmode
  \begin{enumerate}
  \item By Theorem~\ref{thm:antichains:ideals}, every downward closed
    subset of $X$ decomposes into finitely many ideals. Moreover,
    since $\C$ is ideally effective, ideals of $X$ may be effectively
    enumerated. Thus, $M_{\textsf{down}}$ enumerates downward closed
    subsets by enumerating finite subsets of ideals.

  \item Let $D, D' \subseteq X$ be the given downward closed subsets
    prescribed by their ideal decomposition. By
    Prop.~\ref{prop:inclusdansunion}, $D \subseteq D'$ if, and only
    if, for every $I \in \idealdecomp(D)$ there exists $J \in
    \idealdecomp(D')$ such that $I \subseteq J$. Therefore, this test
    can be performed by $M_{\subseteq}$.

  \item Let $D$ be the given downward closed subset prescribed by its
    ideal decomposition. Since $\C$ is ideally effective,
    $M_{\downc{\!\textsf{Post}}}$ can compute $Y_I = \downc{\post(I)}$
    for every $I \in \idealdecomp(D)$. We have,
    \begin{align}
      \downc{\post(D)} = \bigcup_{I \in \idealdecomp(D)} \bigcup_{J
        \in Y_I} J\ .\label{eq:post:d}
    \end{align}
    In order to obtain precisely $\idealdecomp(\downc{\post(D)})$,
    $M_{\downc{\!\textsf{Post}}}$ minimizes (\ref{eq:post:d}) by
    applying Prop.~\ref{prop:inclusdansunion}.

  \item Testing $x \in D$ is equivalent to testing $\downc{x}
    \subseteq D$. $M_{\textsf{memb}}$ obtains the encoding of
    $\downc{x}$ and tests $\downc{x} \subseteq D$ by using
    $M_{\subseteq}$. \qedhere
  \end{enumerate}
\end{proof}

\noindent From Theorem~\ref{thm:cover:wsts:algo} and
Lemma~\ref{lem:down:effectiveness}, we obtain the following result:

\begin{cor}\label{thm:cover:wsts}
  Coverability is decidable for any ideally effective class of WBTS.
\end{cor}

We recall that coverability is undecidable for a large class of WSTS
(hence for WBTS) when computations on ideals are not effective. It was
shown in~\cite{BFM14,BFM16} that coverability is undecidable even for
some post-effective classes of finitely branching WSTS with strong and
strict monotonicity.

As an application of Corollary~\ref{thm:cover:wsts}, we now argue that
vector addition systems with states, a model computationally
equivalent to Petri nets and thus a WSTS, can be extended in a non
articifial way to yield a WBTS that we will call a \emph{weighted
  VASS}. Recall that a \emph{vector addition system with states
  with $d$ counters ($d$-VASS)} is defined as a $\Z^d$-VASS (see
Sect.~\ref{sec:example:wbts}), but where the counters are not allowed
to drop below zero, and where the values of counters are ordered by
the usual componentwise ordering on $\N^d$. We propose to extend VASS
with weights, i.e.\ with additional counters over $\Z$. These counters
may represent, e.g., energy, fuel, time, money, or items of an
inventory, where positive amounts correspond to production or
availability, and negative amounts correspond to consumption or
deficits~\cite{DrosteG07,BokerCHK11,EFLQ13,BGM-atva14,JurdzinskiLS15}. To
the best of our knowledge, such an extension has never been studied
nor introduced. Formally, this new model is defined as follows.

\begin{defi}\label{def:weighted:vass}
  A \emph{weighted $(d, w)$-VASS}, where $d, w \in \N$, is a pair $\V
  = (Q, T)$ such that $Q$ is a finite set of \emph{control states} and
  $T \subseteq Q \times \Z^d \times \Z^w \times Q$ is a finite set of
  \emph{transitions}. A weighted $(d, w)$-VASS induces a transition
  system $(Q \times \N^d \times \Z^w, \srun)$ such that $(p, \vec{u} )
  \srun (q, \vec{v}) \defiff \exists (p, \vec{z}, q) \in T \text{
    s.t. } \vec{v} = \vec{u} + \vec{z}$ and $\vec{v}[1..d] \geq
  \vec{0}$.
\end{defi}

By definition, $d$-VASS and the $\Z^w$-VASS of
Sect.~\ref{sec:example:wbts} are special cases of weighted
VASS. Weighted VASS ordered with the usual componentwise ordering, are
not well-quasi-ordered even with a unique weight counter, and are not
WBTS as soon as they have more than one weight counter. However,
weighted VASS are WBTS when configurations are first ordered according
to the control state and the $d$ first counters, and \emph{then}
lexicographically according to the weights, i.e.\ when ordered under
\begin{align*}
  q(\vec{u}, \vec{v}) \leq q'(\vec{u}', \vec{v}') &\defiff (q = q')
  \land [(\vec{u} <_{\N^d} \vec{u}') \lor ((\vec{u} = \vec{u}') \land
    (\vec{v} \leqlex \vec{v}'))].
\end{align*}
Intuitively, weight counters are ordered according to some priorities,
and act as a tie-breaker among equal control states and
$\N^d$-counters values. Note that with a single weight counter, the
lexicographical ordering is precisely the usual ordering over $\Z$. It
can be shown that $\leq$ does not contain any infinite antichain since
$\N^d$ is well-quasi-ordered and $\leqlex$ does not contain any
infinite antichain.

Moreover, weighted VASS can be shown ideally effective under $\leq$ by
representing every ideal by an $\textsf{FO}(\Z, +, <)$-formula whose
purpose is to answer the membership query in the ideal. First, we
build a formula $\psi_\leq$ such that $\psi_\leq(x, y) \iff x \leq
y$. Then, all properties required for ideal effectiveness to hold are
satisfied as follows:
\begin{itemize}
\item testing whether a formula $\varphi$ encodes an ideal amounts to
  testing whether $\varphi$ is:
  \begin{itemize}
    \item downward closed: $\forall x, y\ [\varphi(x) \land
      \psi_\leq(y, x)] \rightarrow \varphi(y)$,

    \item directed: $\forall x, y\ [(\varphi(x) \land \varphi(y))
      \rightarrow \exists z\ (\psi_\leq(x, z) \land \psi_\leq(y, z)
      \land \varphi(z))]$;
  \end{itemize}

\item the ideal $\downc{x}$ can be represented by $\varphi_{x}(y)
  \defeq \psi_\leq(y, x)$;

\item inclusion of ideals $I$ and $J$ represented respectively by
  formulas $\varphi_I$ and $\varphi_J$ can be decided by testing
  $\forall x\ \varphi_I(x) \rightarrow \varphi_J(x)$;

\item given a formula $\varphi_I$ for an ideal $I$, the set
  $\downc{\spost(I)}$ can be represented by $\{\varphi_{\spost, t} : t
  \in T\}$ where $\varphi_{\spost, t}(x) \defeq \exists y,
  z\ [\psi_\leq(x, z) \land \varphi_I(y) \land \varphi_t(y, z)]$ and
  $\varphi_t(y, z)$ holds if and only if $y$ leads to $z$ under
  transition $t$.
\end{itemize}
Therefore, weighted VASS form a post-effective and ideally effective
class of WBTS under $\leq$, and by Corollary~\ref{thm:cover:wsts},
coverability is decidable for this model.

It is worth mentioning that weighted VASS are \emph{not} WBST under
the following similar but different ordering:
\begin{align*}
  q(\vec{u}, \vec{v}) \leq' q'(\vec{u}', \vec{v}') \defiff (q = q')
  \land (\vec{u} \leq_{\N^d} \vec{u}') \land (\vec{v} \leqlex
  \vec{v}').
\end{align*}
Indeed, $\{q(n, -n) : n \in \N\}$ is an infinite antichain for $\leq'$
when $d = 1$ and $w = 1$.

\section{Termination and boundedness}\label{sect:termination}

The termination and boundedness problems are respectively defined as
follows: on input an ordered transition system $\S = (X, \srun, \leq)$
and a state $x$, determine respectively whether
\begin{itemize}
  \item $\S$ terminates from $x$, i.e.\ there is no infinite sequence
    $x_1, x_2, \dots \in X$ such that $x \srun x_1 \srun x_2 \srun
    \dots$;
  \item $\S$ is bounded from $x$, i.e.\ $\sposts(x)$ is finite.
\end{itemize}

These two problems are undecidable in general, even for some classes
of finitely branching (non effective) WSTS. However, they are decidable under
reasonable monotonicity and effectiveness hypotheses~(see
e.g.~\cite{Finkel&Schnoebelen:01}). We observe that under these
hypotheses, termination and boundedness do not remain decidable for
WBTS. Hence WSTS and WBTS behave differently with respect to the
decidabilities of their termination and boundedness problems.

\begin{lem}
  There exists a post-effective class of finitely branching
  WBTS, with strong and strict monotonicity, and partial
  ordering, for which termination and boundedness are undecidable.
\end{lem}

\begin{proof}
  We give a reduction from the halting problem. Let $\tmi$ be the
  $i^\text{th}$ Turing machine in a classical enumeration. Let $\S_i
  \defeq (X, \srun, \leq)$ be the ordered transition system defined by
  $X \defeq \{0\} \cup (\Z-\N) = \{0, -1, -2, \dots\}$ and $$x \srun x
  - 1 \defiff \tmi \text{ does not halt on its encoding in $|x|$ steps
    or less}\ .$$ Let $\C \defeq \{\S_i : i \in \N\}$. We first show
  that $\C$ is a class of WBTS as described in the proposition. Let $i
  \in \N$. Since $|\post_{\S_i}(x)| \leq 1$ for every $x \in X$,
  $\S_i$ is finitely branching. Moreover, $\C$ is post-effective since
  testing $x \srun y$ only requires executing a Turing machine for a
  finite number of steps. Because $X$ is a partially ordered set
  without any infinite antichain, it remains to prove strong and
  strict monotonicity. Let $x, y, x' \in X$ be such that $x \srun y$
  in $\S_i$ and $x' > x$. By definition of $\srun$, $y = x-1$ and
  $\tmi$ does not halt in $|x|$ steps or less. Therefore, by $|x'| <
  |x|$, $\tmi$ does not halt in $|x'|$ steps or less, hence $x' \srun
  y'$ where $y' = x'-1 > x-1 = y$.

  Now, we note that there exists an infinite sequence $0, x_1, x_2,
  \dots \in X$ such that $0 \srun x_1 \srun x_2 \srun \dots$ in $\S_i$
  if, and only if, $\tmi$ does not halt, if and only if, $\sposts(0)$
  is infinite. Therefore, we conclude that termination and boundedness
  are both undecidable.
\end{proof}

Despite these negative results, we may exhibit a subclass of WBTS for
which termination and boundedness are decidable. Recall that the
reachability tree from an initial state $x_0$ in a transition system
$\mathcal{S}$ is a tree rooted at $x_0$ and having an edge $(x, y)$
for each pair of states $x, y$ such that $x \srun{} y$. Analogous to
the finite reachability tree for WSTS~\cite{Finkel&Schnoebelen:01},
which is obtained from truncation of the reachability tree, we define
the antichain tree that will provide algorithms for termination and
boundedness. Informally, whereas the criterion for truncating a branch
at a node labelled $x_j$ in the reachability tree is the occurrence of
an ancestor labelled $x_i$ with $x_i\leq x_j$, the criterion for
truncation in the antichain tree will be the weaker condition on $x_i$
that \emph{either} $x_i\leq x_j$ or $x_j\leq x_i$:

\begin{defi}[Antichain tree]
  Let $\S = (X, \srun, \leq)$ be a WBTS, and let $x_0 \in
  X$. The \emph{antichain tree} of $\S$ from the initial state $x_0$
  is a partial reachability tree $\AT(\S, x_0)$ with root $c_0$
  labelled $x_0$ that is defined and built as follows. For every $x
  \in \spost(x_0)$ we add a child labelled $x$ to $c_0$. The tree is
  then built iteratively in the following way. Only an unmarked node
  $c$ labelled $x$ is picked:
  \begin{itemize}
  \item if $c$ has an ancestor $c'$ labelled $x'$ such that $x' \leq
    x$ or $x \leq x'$, we mark $c$
  \item otherwise, we mark $c$ and for every $y \in \spost(x)$ we add
    a child labelled $y$ to $c$.
  \end{itemize}
\end{defi}

We observe that each path of the antichain tree is a prefix of a path
of the finite reachability tree, which is
finite~\cite[Lemma~4.2]{Finkel&Schnoebelen:01}, hence the antichain
tree is also finite. More formally:

\begin{lem}\label{lem:at}
  The antichain tree is finite and computable for finitely branching
  and post-effective WBTS.
\end{lem}

\begin{proof}
  Suppose that $\AT(\S, x_0)$ is infinite. As $\S$ is finitely
  branching, by K\"{o}nig's Lemma, there is an infinite branch $c_0
  \srun_{\AT} c_1 \srun_{\AT} \dots$ in this tree labelled by the
  following infinite sequence: $x_0, x_1, \dots$ Since $\leq$ is
  without infinite antichains, there is a least $j$ for which some $i
  < j$ satisfies $x_i \leq x_j$ or $x_j \leq x_i$. But then the branch
  would have been truncated at $c_j$ or at $c_i$ and this is a
  contradiction, hence $\AT(\S, x_0)$ is finite. The tree is
  computable since $\S$ is post-effective.
\end{proof}

Let us state a useful lemma.

\begin{lem}\label{lem:maxpath}
  Any path in the reachability tree of a WBTS $\S$ from $x_0$
  has a finite prefix labelling a maximal path in the antichain tree
  $\AT(\S, x_0)$.
\end{lem}

The proof of Lemma~\ref{lem:termtool} is a (self-contained) adaptation
of the proof of \cite[Prop.~4.5]{Finkel&Schnoebelen:01}.

\begin{lem}\label{lem:termtool}
  Let $\S = (X, \srun, \leq)$ be a finitely branching WBTS with
  upward and downward transitive monotonicity. Then $\S$ \emph{does
    not terminate} from $x_0$ if, and only if, there is a path $c_0
  \srun_{\AT} c_1 \srun_{\AT} \dots \srun_{\AT} c_i \srun_{\AT} \dots
  \srun_{\AT} c_j$ in $\AT(\S, x_0)$ with labels $x_0, x_1, \dots,
  x_j$ such that $x_i \leq x_j$ or $x_j \leq x_i$.
\end{lem}

\begin{proof}
  \emph{Only if}. Suppose that an infinite run $x_0 \srun x_1 \srun
  \dots$ exists in $\S$. By Lemma~\ref{lem:maxpath}, a maximal path
  $c_0 \srun_{\AT} c_1 \srun_{\AT} \dots \srun_{\AT} c_j$ with labels
  $x_0, x_1, \dots, x_j$ exists in $\AT(\S, x_0)$. Since this path is
  maximal, it ought to have been the presence of some $i < j$ with
  $x_i \leq x_j$ or $x_j \leq x_i$ that caused the truncation.

  \noindent \emph{If}. Suppose that a path $c_0 \srun_{\AT} c_1
  \srun_{\AT} \dots \srun_{\AT} c_i \srun_{\AT} \dots \srun_{\AT} c_j$
  with comparable labels $x_i$ and $x_j$ exists in $\AT(\S, x_0)$.
  Then a run $x_0 \srunk{*} x_i \srun x_{i+1} \srunk{*} x_j$ is
  possible in $\S$. If $x_j \leq x_i$, then by downward transitive
  monotonicity, there exists $x_{j+1} \leq x_{i+1}$ such that $x_j
  \srunk{+} x_{j+1}$. By induction, for every $m > j$, there exist
  $x_{j+1}, x_{j+2}, \dots, x_{j+m}$ such that $x_0 \srunk{*} x_i
  \srunk{+} x_j \srunk{+} x_{j+1}\srunk{+} x_{j+2}\srunk{+} \cdots
  \srunk{+} x_{j+m}$. But then, by applying K\"{o}nig's Lemma to the
  finitely branching reachability tree of $\S$, we note that $\S$ does
  not terminate from $x_0$. The case $x_i \leq x_j$ is treated
  similarly, using the upward transitive monotonicity.
\end{proof}

The proof of Lemma~\ref{lem:boundedtool}
adapts~\cite[Prop.\ 4.10]{Finkel&Schnoebelen:01} and strengthens it in
that both transitive and strict monotonicity are required there, while
only strict monotonicity is required here.

\begin{lem}\label{lem:boundedtool}
  Let $\S = (X,\srun,\leq)$ be a finitely branching WBTS with
  upward and downward strict monotonicity and such that $\leq$ is a
  partial ordering. Then $\S$ \emph{is not bounded from} $x_0$ if, and
  only if, there is a path $c_0 \srun_{\AT} c_1 \srun_{\AT} \dots
  \srun_{\AT} c_i \srun_{\AT} \dots \srun_{\AT} c_j$ in $\AT(\S, x_0)$
  with labels $x_0, x_1, \ldots, x_j$ such that $x_i < x_j$ or $x_j <
  x_i$.
\end{lem}

\begin{proof}
  \emph{Only if}. Suppose that $\sposts(x_0)$ is infinite.  Consider
  the reachability tree defined from cycle-free runs (hence runs with
  no repeated states) from $x_0$ in $\S$. By K\"{o}nig's lemma applied
  to this finitely branching tree, some such run $x_0 \srun x_1 \srun
  \dots$ in $\S$ is infinite. As in the proof of
  Lemma~\ref{lem:termtool}, Lemma~\ref{lem:maxpath} implies the
  existence in $\AT(S,x_0)$ of a path $c_0 \srun_{\AT} c_1 \srun_{\AT} \dots
  \srun_{\AT} c_i \srun_{\AT} \dots \srun_{\AT} c_j$ with labels $x_0,
  x_1, \dots, x_j$ such that either $x_i \leq x_j$ or $x_j \leq x_i$.
  Being distinct and comparable in a partial order, $x_i$ and $x_j$
  satisfy $x_i < x_j$ or $x_j < x_i$, as required.

  \noindent\emph{If}. Suppose that there exists a path $c_0
  \srun_{\AT} c_1 \srun_{\AT} \cdots \srun_{\AT} c_i \srun_{\AT}
  \cdots \srun_{\AT} c_j$ such that $x_i < x_j$ or $x_j < x_i$ exists
  in $\AT(\S, x_0)$. Then a run
  $$x_0 \srunk{*} x_i \srunk{k_0} x_j$$ is possible in $\S$ for the
  appropriate $k_0 > 0$. If $x_j < x_i$, then by $k_0$ applications of
  strict downward monotonicity, there exists $y_1 < x_j$ such that
  $$y_0 \defeq x_j \srunk{*} y_1.$$ Since $y_1 < y_0$, $y_0
  \srunk{k_1} y_1$ for some $k_1 > 0$. Hence the argument can be
  repeated to exhibit an infinite descending chain $y_0 > y_1 > y_2>
  \dots$ such that
  $$x_0 \srunk{*} x_i \srunk{*} y_0 \srunk{*} y_1 \srunk{*} y_2
  \cdots\ .$$ Hence $\sposts(x_0)$ is infinite. The case $x_i < x_j$
  is treated similarly, using upward strict monotonicity.
\end{proof}

The following holds:

\begin{thm}\label{thm:termplusbound}
\leavevmode
\begin{itemize}
\item Termination is decidable for any post-effective class of
  finitely branching WBTS with upward and downward transitive
  monotonicity.
\item Boundedness is decidable for any post-effective class of
  finitely branching WBTS with upward and downward strict
  monotonicity and partial ordering.
\end{itemize}
\end{thm}

\begin{proof}
  Given a WBTS $\S = (X, \srun, \leq)$ and $x_0 \in X$, both
  the termination and the boundedness algorithm begin with the
  computation of $\AT(\S, x_0)$, doable by Lemma~\ref{lem:at}. Since
  $\AT(\S, x_0)$ is finite, the algorithm for termination can proceed
  to test the condition of Lemma~\ref{lem:termtool} and the algorithm
  for boundedness the condition of Lemma~\ref{lem:boundedtool}.
\end{proof}

\begin{rem}
  Under the hypotheses of Theorem~\ref{thm:termplusbound}, boundedness
  is decidable even when WBTS are infinitely branching. Indeed, it
  suffices in this case to add to the construction of the antichain
  tree the rule that a branch is further truncated when a node $x$
  such that $|\post(x)| = \infty$ is encountered. Recall that by
  definition of post-effectiveness, such an occurrence can be
  detected. Moreover, any such occurrence in the antichain tree
  implies unboundedness.
\end{rem}

\section{Conclusion}\label{sect:conc}

In this work we have noted that well-foundedness of the quasi-ordering
traditionally used to define a WSTS is not required for the purpose of
deciding coverability. Accordingly, we have defined WBTS by relaxing
the conditions on the ordering so as to only require the absence of
infinite antichains.

As proof of concept, we have introduced an extension of vector
addition systems called weighted $(d,w)$-VASS.  Weighted $(d,w)$-VASS
operate on their $\N^d$ component as normal VASS and they operate
without guards on a new $\Z^w$ component ordered by lexicographically
extending the usual order on $\Z$.  The resulting model is a WBTS that
is not a WSTS.  From studying the ideal structure of downward closed
subsets of $\Z^w$ under the latter ordering, we deduced that all
necessary effectiveness conditions hold for a forward algorithm to be
able to decide coverability for weighted $(d,w)$-VASS.  More
generally, this forward algorithm was shown able to decide
coverability for any WBTS that possesses the ``ideally effective''
property.

To delimit the picture, we have further shown that, unlike in the
well-studied case of WSTS, the termination and the boundedness
problems for WBTS become undecidable.  On the other hand, appropriate
downward and upward monotonicity conditions were shown to bring back
decidability for these problems.

As future work directions, other WBTS and orderings could be
studied. For example, the lexicographical ordering on words over a
finite alphabet could be used in lieu of $\Z^w$ as the weight domain
of WBTS and weighted VASS.  Beyond studying the ideal structure of
$\Sigma^*$ under this ordering for its own sake, it is conceivable
that models of practical use in verification might use such an
ordering for the purpose of modelling priorities.  Given the recent
focus on the complexity of VASS problems~\cite{Schmitz16},
investigating complexity questions for weighted VASS and specific WBTS
would certainly be worthwhile.

But the final take-home message of this paper might be that, as the
need arises, new models weaker than the WSTS can now be defined with
some hope for usability.

\section*{Acknowledgements}

The second author would like to thank Raphaël Carroy, Mirna Dzamonja,
Yan Pequignot and Maurice Pouzet for discussions on
Theorem~\ref{thm:antichains:ideals} at the Dagstuhl Seminar 16031 on
well quasi-orders in computer science held in January 2016. We would
also like to thank Philippe Schnoebelen for his valuable comments on
an early version of this paper and for sharing a draft version
of~\cite{GKNP16}, a paper in preparation with Jean Goubault-Larrecq,
Prateek Karandikar and K. Narayan Kumar whom we also thank. We thank
Laurent Doyen and Yaron Welner as well for helpful discussions.

\bibliographystyle{alpha}
\bibliography{references}

\end{document}